\documentclass[a4paper]{jpconf}
\usepackage{graphicx}
\usepackage{iopams}

\newtheorem{lemma}{Lemma}

\newtheorem{proposition}{Proposition}

\newenvironment{proof}[1][Proof]{\textbf{#1.} }{\ \rule{0.5em}{0.5em}}

\begin{document}
\title{Obtainment of internal labelling operators as broken Casimir operators by
means of contractions related to reduction chains in semisimple
Lie algebras}

\author{R. Campoamor-Stursberg}

\address{Dpto. Geometr\'{\i}a y Topolog\'{\i}a\\Fac. CC. Matem\'aticas\\
Universidad Complutense de Madrid\\Plaza de Ciencias, 3\\E-28040
Madrid, Spain}

\ead{rutwig@pdi.ucm.es}

\begin{abstract}
We show that the In\"on\"u-Wigner contraction naturally associated
to a reduction chain $\frak{s}\supset \frak{s}^{\prime}$ of
semisimple Lie algebras induces a decomposition of the Casimir
operators into homogeneous polynomials, the terms of which can be
used to obtain additional mutually commuting missing label
operators for this reduction. The adjunction of these scalars that
are no more invariants of the contraction allow to solve the
missing label problem for those reductions where the contraction
provides an insufficient number of labelling operators.

\end{abstract}

\section{Introduction}

One of the main applications of group theoretical methods to
physical problems is related to classification schemes, where
irreducible representations of a Lie  group have to be decomposed
into irreducible representations of a certain subgroup appearing
in some relevant reduction chain
\begin{equation}
\left|
\begin{array}
[c]{ccccccc}%
 \frak{s} & \supset &  \frak{s}^{\prime} & \supset &  \frak{s}^{\prime\prime} & ...\\
\downarrow &  & \downarrow &  & \downarrow & \\
 \left[  \lambda\right]   &  &  \left[
\lambda^{\prime}\right] & &  \left[ \lambda^{\prime\prime}\right]
& ...
\end{array}
\right\rangle. \label{Red1}
\end{equation}
This is the case for dynamical symmetries used for example in
nuclear physics, where one objective of the algebraic model is to
describe the Hamiltonian (or mass operator in the relativistic
frame) as a function of the Casimir operators of the chain
elements. The corresponding energy formulae can the easily deduced
from the expectation values in the reduced representations. As
example, the Gell-Mann-Okubo mass formula can be derived using
this ansatz \cite{Ok}. In many situations, the labels obtained
from the reduction (\ref{Red1}) are sufficient to solve the
problem, e.g., of we require multiplicity free reductions, as used
in $SU(N)$ tumbling gauge models \cite{Sua} or the interacting
boson model \cite{Ia}. However, often the subgroup does not
provide a sufficient number of labels to specify the basis states
unambigously, due to multiplicities greater than one for induced
representations. This turns out to be the rule for non-canonical
embeddings and generic representations of $\frak{g}$. For some
special types, like totally symmetric or anti-symmetric
representations, additional labels are not necessary to solve the
problem, and the degeneracies can be solved directly with the
available operators.

\medskip

Many procedures methods have been developed to solve the so-called
missing label problem (short MLP), from specific construction of
states for the reduction chain to the formal construction of all
possible labelling operators using enveloping algebras \cite{El}.
The latter procedure allows in theory to find the most general
labelling operator, although the effective computation of
admissible operators is rather cumbersome. In addition, there is
no general criterion to compute the number of operators necessary
to generate integrity bases in enveloping algebras.

\medskip

In the mid seventies, Peccia and Sharp \cite{Pe} considered an
analytic approach to the MLP based on the method used to compute
the generalized Casimir invariants of Lie algebras in the
commutative frame. This method is very close to the interpretation
of Casimir operators as functions that are constant on co-adjoint
orbits, and consists essentially on a restriction of this problem.
The labelling operators are interpreted as solutions of a certain
subsystem of partial differential equations corresponding to the
embedded subalgebra. Subgroup scalars are therefore the
differential invariants of the subalgebra in a specific
realization, and their symmetrization serves to separate
multiplicities of reduced representations. Although this approach
is, in principle, computationally easier than the pure algebraic
approach using operators in the enveloping algebra, the
integration of such differential equations is far from being
trivial. Moreover, the orthogonality conditions required to the
labelling operators, in order to avoid undesirable interactions,
must be solved by pure algebraic methods in most cases.

\medskip

Most of the MLP considered in the literature have been solved or
studied from an algebraic or analytic point of view, generally
looking for solutions of lowest degrees. In appearance, no
attention has been paid to the properties that the embedding
imply.\footnote{As known, non-equivalent embeddings lead to
different branching rules, and therefore to different
classification schemes. This is the difference between the Wigner
supermultiplet (nuclear LS coupling) and the strange-spin
multiplet structure of $\frak{su}(4)$ \cite{Is}.} The embedding
$\frak{s}\supset \frak{s}^{\prime}$ is conditioned by physical
reasons, that is, the choice of the embedding class corresponds to
some specific coupling scheme or some relevant internal property
that must be emphasized (like angular momentum). Since the
embedding determines the branching rules for irreducible
representations, it should be expected that labelling operators
needed to solve multiplicities are, in some manner, codified by
the properties of the symmetry breaking determined by the
reduction chain.\newline When dealing with the MLP algebraically
or analytically, it is not immediately clear to which extent we
are using the properties of this embedding or the branching rules.
Moreover, we can ask whether the obtained labelling operators have
some intrinsic or physical meaning at all. Formulated in another
way: are the labelling operators of the MLP completely determined
by the reduction $\frak{s}\supset \frak{s}^{\prime}$ (and
therefore, by the underlying physics), or are they the result of a
formal algebraic/analytic manipulation?

\medskip

In \cite{C72}, the missing label problem was approached from a
quite general point of view, but having in mind the important fact
observed in \cite{Vi} that symmetry breaking and contractions of
Lie algebras have many points in common. This actually is deeply
related to the characterization of inhomogeneous algebras obtained
from contractions of semisimple algebras \cite{C49}. Therefore,
any reduction chain $\frak{s}\supset \frak{s}^{\prime}$ naturally
induces a contraction of Lie algebras. It is therefore natural to
ask if the operators provided by the contraction $\frak{g}$ (the
so-called contracted Casimir operators) can be used to solve the
MLP completely. The idea, in a different form, had been used
previously for angular momentum subalgebra, and can be seen
clearly in the so-called rotor expansion method \cite{El}.

\medskip

The main goal of the general contraction ansatz in labelling
problems can be resumed in the following points:

\begin{enumerate}

\item Find a procedure to solve the  MLP  using explicitly the
properties of the embedding  $\frak{s}\supset \frak{s}^{\prime}$
[breaking symmetry $\leftrightarrow$ contractions], that is,
without invoking formal operators.

\item Justify a  natural  choice of labelling operators as
``broken Casimir operators".

\item Find a phenomenological explanation for the non-integer
expectation values of labelling operators.

\end{enumerate}

In this first development, only the contracted invariants were
used to generate labelling operators. This approach was however
sufficient to solve many physically relevant missing label
problems, and the result were in perfect harmony with those
obtained using other techniques. It was also observed that the
method can fail to find a sufficient number of missing operators
whenever the identity
$\mathcal{N}(\frak{g})=\mathcal{N}(\frak{s})=n$ is satisfied. The
failure is essentially a consequence of an insufficient number of
invariants in the contraction.

\medskip

The aim of this work is to further develop the contraction
procedure, but using not only the contracted invariants, but a
certain decomposition induced in the Casimir operators of the
contracted algebra, which turn out to be subgroup scalars but no
more invariants of the contraction. With this decomposition, we
are able to surmount the difficulty for cases with insufficient
number of contracted operators. A more interesting consequence of
this fact is the possibility of explaining  the existence of
missing label operators of the same degree, as they have already
been constructed in the algebraic frame.

\section{Classical Casimir operators}

Given a presentation $ \frak{s}=\left\{X_{1},..,X_{n}\; |\;
\left[X_{i},X_{j}\right]=C_{ij}^{k}X_{k}\right\}$ of a Lie algebra
$\frak{s}$ in terms of generators and commutation relations, we
are interested in (polynomial) operators
$C_{p}=\alpha^{i_{1}..i_{p}}X_{i_{1}}..X_{i_{p}}$ in the
generators of $\frak{s}$ such that the constraint $
\left[X_{i},C_{p}\right]=0$,\; ($i=1..n$) is satisfied. Such an
operator necessarily lies in the centre of the enveloping algebra
of $\frak{s}$, and is traditionally referred to as Casimir
operator. However, in many dynamical problems, the relevant
invariant functions are not polynomials, but rational or even
trascendental functions (e.g. the inhomogeneous Weyl group).
Therefore the approach with the universal enveloping algebra has
to be generalized in order to cover arbitrary Lie groups. The most
widely used method is the analytical realization. The generators
of the Lie algebra $\frak{s}$ are realized by the differential
operators $\widehat{X}_{i}= C_{ij}^{k} x_{k}
\frac{\partial}{\partial x_{j}}$, where the $x_{i}$ are commuting
variables associated to each generator $X_{i}$. In this approach,
a function $F\in C^{\infty}(\frak{s}^{\prime})$ is an invariant of
$\frak{s}$ if and only if it satisfies the system of PDEs
\begin{equation}
\widehat{X}_{i}(F)=0, \quad  i=1..n.\label{sys}
\end{equation}
Using the symmetrization map
$Sym(x_{1}^{a_{1}}..x_{p}^{a_{p}})=\frac{1}{k!} \sum_{\sigma\in
S_{p}} x_{\sigma(1)}^{a_{1}}...x_{\sigma(p)}^{a_{p}}$, we recover
the Casimir operator for polynomial solutions. With this
analytical ansatz, it is easily seen that the number of
independent solutions is
\begin{equation}
\mathcal{N}\left( \frak{s}\right) =\dim\frak{s}-\;
\mathrm{rank\,}\left[ C_{ij}^{k}x_{k}\right].
\end{equation}

As already commented, one of the main applications of Casimir
operators of Lie algebras is the labelling of irreducible
representations. In a more general approach, irreducible
representations of a Lie algebra $\frak{g}$ can be labelled using
the eigenvalues of its generalized Casimir invariants \cite{Pe}.
For each representation, the number of internal labels needed is
given by
\begin{equation}
i=\frac{1}{2}(\dim \frak{g}- \mathcal{N}(\frak{g})).
\end{equation}
For the special case of semisimple Lie algebras, this number is
deeply related to the number of positive roots of its
complexification. If we use some subalgebra $\frak{g}\supset
\frak{h}$ to label the basis states of $\frak{g}$, the embedding
provides $\frac{1}{2}(\dim
\frak{h}+\mathcal{N}(\frak{h}))+l^{\prime}$ labels, where
$l^{\prime}$ is the number of invariants of $\frak{g}$ that depend
only on variables of the subalgebra $\frak{h}$ \cite{Pe}. In
general, this is still not sufficient to separate multiplicities
of induced representations, so that we need  to find
\begin{equation}
n=\frac{1}{2}\left(
\dim\frak{g}-\mathcal{N}(\frak{g})-\dim\frak{h}-\mathcal{N}(\frak{h})\right)+l^{\prime}
\label{ML}
\end{equation}
additional operators, which are commonly called missing label
operators. The total number of available operators of this kind is
twice the number of needed labels, $m=2n$. For the case $n>1$, it
remains the problem of determining a set of $n$ mutually commuting
operators, as commented before. These operators can be seen to be
subgroup scalars, so that the analytical approach is a practical
method to find the labelling operators.\footnote{In many cases,
the degeneracies of the reduction chain can be solved without
using the subgroup scalars, or are determined by a specifical
ansatz adapted to the involved groups.} Considering the
realization of $\frak{s}\supset \frak{s}^{\prime}$ by differential
operators indicated in (\ref{sys}), we restrict to the PDEs
corresponding to subgroup generators
\begin{equation}
\widehat{X}_{i}= C_{ij}^{k} x_{k} \frac{\partial}{\partial
x_{j}},\quad 1\leq i\leq \dim\frak{s}^{\prime}. \label{SML}
\end{equation}
Solutions to this system $f(\frak{s}^{\prime})$ reproduce the
differential invariants in $\dim \frak{s}$ dimensions. The total
number of solutions of the latter system is given by:
\begin{equation}
\mathcal{N}(f(\frak{s}^{\prime}))=m+\mathcal{N}(\frak{s})+\mathcal{N}(
\frak{s}^{\prime})-l^{\prime}. \label{ML2}
\end{equation}
We observe that (\ref{ML2}) refers to the number of functionally
independent solutions. Integrity bases, constrained by the coarser
algebraic independence, will generally have much more elements
than (\ref{ML2}), and no general procedure is known to compute its
dimension. In addition, to be useful as labelling operators, two
(symmetrized) solutions $F_{1},F_{2}$ of system (\ref{SML}) must
satisfy the orthogonality condition $ \left[F_{1},F_{2}\right]=0$
and commute with the Casimir operators $C_{i}$ and $D_{j}$ of
$\frak{s}$ and the subalgebra $\frak{s}^{\prime}$.  Therefore the
set of commuting operators $\left\{C_{i},D_{j},F_{k}\right\}$
serves to label the states unambiguously.

\medskip

It was proved in \cite{C72} that any reduction chain
$\frak{s}\supset \frak{s}^{\prime}$  is naturally associated to an
In\"on\"u-Wigner contraction $\frak{s}\rightsquigarrow
\frak{g}=\frak{s}^{\prime}\overrightarrow{\oplus}_{R}\left(
\dim\frak{s}-\dim\frak{s}^{\prime}\right) L_{1}$, where the
representation $R$ of $\frak{s}^{\prime}$ satisfies the
constraint\footnote{Actually, the branching rule depends on a
numerical index $j_{f}$ that characterizes the embedding class.}
\begin{equation}
{\rm ad} \frak{s}= {\rm ad} \frak{s}^{\prime}\oplus R.\label{IRD}
\end{equation}
The contraction is easily seen to be defined by the non-singular
transformations
\begin{equation}
\Phi_{t}\left(  X_{i}\right)  =\left\{
\begin{array}
[c]{cc}%
X_{i}, & 1\leq i\leq s\\
\frac{1}{t}X_{i}, & s+1\leq i\leq n
\end{array}
\right.  ,\label{TB}
\end{equation}
where  $\left\{ X_{1},..,X_{s},X_{s+1},..,X_{n}\right\} $ is a
basis of $\frak{s}$  such that  $\left\{  X_{1},..,X_{s}\right\} $
generates the subalgebra $\frak{s}^{\prime}$  and $\left\{
X_{s+1},..,X_{n}\right\} $  spans the representation space $R$. If
both $\frak{s}$ and $\frak{s}^{\prime}$ are semisimple, the
contraction is isomorphic to an inhomogeneous Lie algebra with
Levi decomposition
$\frak{g}=\frak{s}^{\prime}\overrightarrow{\oplus}_{R}\left(
\dim\frak{s}-\dim\frak{s}^{\prime}\right) L_{1}$. Since it
satisfies the condition $\left[\frak{g},\frak{g}\right]=\frak{g}$,
this algebra admits a fundamental basis of invariants consisting
of Casimir operators. Moreover, by  the contraction we have the
inequality $\mathcal{N}(\frak{s})\leq \mathcal{N}(\frak{g})$. The
system of PDEs corresponding to this contraction can be divided
into two parts:

\begin{eqnarray}
\fl \widehat{X}_{i}F=C_{ij}^{k}x_{k}\frac{\partial F}{\partial
x_{j}}=0,& \quad 1\leq i\leq s, \label{KS1}\\
\fl \widehat{X}_{s+i}F=C_{s+i,j}^{s+k}x_{s+k}\frac{\partial
F}{\partial x_{j}}=0, & 1\leq i,k\leq n-s, 1\leq j\leq
s.\label{KS2}
\end{eqnarray}
The subsystem (\ref{KS1}) corresponds to the generators of
$\frak{s}^{\prime}$ realized as subalgebra of $\frak{s}$, while
the remaining equations (\ref{KS2}) describe the representation.
Written in matrix form, the system is given by
\[
\left(
\begin{array}
[l]{ccccll}
0 & ... & C_{1s}^{k}x_{k} &  C_{1,s+1}^{k}x_{k} & ...
&  C_{1,n}^{k}x_{k}\\
\vdots & & \vdots & \vdots & & \vdots\\
 -C_{1s}^{k}x_{k} & ... & 0 &  C_{s,s+1}^{k}x_{k} & ... &
 C_{s,n}^{k}x_{k}\\
 -C_{s,s+1}^{k}x_{k} & ... &  -C_{s,s+1}^{k}x_{k} &  0 &...&
 0\\
\vdots & & \vdots & \vdots & & \vdots\\
 -C_{1n}^{k}x_{k}& ...&  -C_{s,n}^{k}x_{k}&  0& ... & 0
\end{array}
\right)\left(
\begin{array}
[c]{c}
\partial_{x_{1}}F\\
\vdots\\
\partial_{x_{s}}F\\
\partial_{x_{s+1}}F\\
\vdots\\
\partial_{x_{n}}F\\
\end{array}\right)=0.
\]
Since the first $s$ first rows reproduce exactly the system
 of PDEs needed to compute the missing label operators, we
 conclude that any invariant of $ \frak{g}$ is a candidate for missing label
operator. In view of this situation, the following questions arise
naturally:
\begin{enumerate}
\item  Do polynomial functions of the invariants of these algebras
suffice to determine $n$  mutually orthogonal  missing label
operators?
\item  Can all  available  operators found by this
procedure?
\end{enumerate}

The answer to the first equation is in the affirmative for those
cases where the contraction provides a number of independent
invariants exceeding the number of needed labelling operators. It
fails when these two quantities coincide, which suggests that the
procedure has to be refined. The answer to the second question is
generally in the negative (see e.g. the algebras with one missing
label), although it cannot be excluded that with the refinement
proposed in this work we are able to recover a complete set of
independent labelling operators for some special reductions. In
most cases, only half of the available operators should be
expected, since all operators obtained are the result, in some
sense, of ``breaking" the original Casimir operators. Whether a
further refinement allows to obtain additional labelling operators
that are independent remains for the moment an unanswered
question.

\section{Decomposition of Casimir operators}

In this section we prove that the contraction induced by the
reduction chain induces a decomposition of the corresponding
Casimir operators of $\frak{s}$, which allow, among other
properties, to determine the invariants of the contraction
$\frak{g}$. However, other terms are also relevant for the missing
label problem, and constitute the solution to the problem pointed
out in \cite{C72} when the number of invariants of the contraction
is not sufficient. These additional terms do not constitute
invariants of the contraction, and where therefore not considered
in \cite{C72}.

We briefly recall the definition of contracted invariants. Since
classical In\"on\"u-Wigner contractions are the only type of
contractions needed for the labelling problem, we can restrict
ourselves to this case \cite{IW}. Let
$C_{p}(X_{1},...,X_{n})=\alpha^{i_{1}...i_{p}}X_{i_{1}}...X_{i_{p}}$
be a $p^{th}$-order Casimir operator of $\frak{s}$. Then the
transformed invariant takes the form
\begin{equation}
 F(\Phi_{t}(X_{1}),..,\Phi_{t}(X_{n})) =
t^{n_{i_{1}}+...+n_{i_{p}}}\alpha^{i_{1}...i_{p}}X_{i_{1}}...X_{i_{p}},
\end{equation}
where $n_{i_{j}}=0,1$. Taking the maximal power in $t$,
\begin{equation}
 M  = \max \left\{n_{i_{1}}+...+n_{i_{p}}\quad |\quad
\alpha^{i_{1}..i_{p}}\neq 0\right\},
\end{equation}
the limit
\begin{eqnarray*}
F^{\prime}(X_{1},..,X_{n}) =  \lim_{t\rightarrow \infty}
t^{-M}F(\Phi_{t}(X_{1}),...,\Phi_{t}(X_{n}))\\
 =  \sum_{n_{i_{1}}+...+n_{i_{p}}=M}
\alpha^{i_{1}...i_{p}}X_{i_{1}}...X_{i_{p}}
\end{eqnarray*}
provides a Casimir operator of degree $p$ of the contraction
$\frak{g}^{\prime}$. Now, instead of extracting only the term with
the highest power of $t$, we consider the whole decomposition
\begin{equation}
C_{p}=t^{M} C_{p}^{\prime}+ \sum_{\alpha}t^{\alpha} \Phi_{\alpha}
+  \Phi_{0}, \label{COZ}
\end{equation}
where $\alpha <M\leq p$ and $\Phi_{0}$ is a function of the
Casimir operators of the subalgebra $\frak{s}^{\prime}$ (these
generators have not been re-scaled). It is straightforward to
verify that $C_{p}^{\prime}$ is not only an invariant of the
contraction $\frak{g}$, but also a solution to the
MLP.\footnote{For $t=1$, equation (\ref{COZ}) shows how the
Casimir operator decomposes into homogeneous polynomials in the
variables of the subalgebra and the complementary space over the
original basis.} This first term was central to the argumentation
in \cite{C72}, and allowed to obtain commuting sets of labelling
operators. However, the remaining terms can also be individually
considered as candidates for labelling operators, as states the
following

\begin{proposition}
The functions $\Phi_{\alpha}$ are solutions of the missing label
problem, that is, they satisfy the system
\begin{equation}
 \widehat{X}_{i}\Phi_{\alpha} = C_{ij}^{k}x_{k}\frac{\partial
\Phi_{\alpha}}{\partial x_{j}} = 0, \quad 1\leq i\leq s.
\label{ML1}
\end{equation}
\end{proposition}

\begin{proof}
First of all, the decomposition (\ref{COZ}) tells that the Casimir
operator $C_{p}$ can be rewritten as a sum of homogeneous
polynomials $C_{p}^{\prime}, \Phi_{\alpha}$ with the property that
$C_{p}^{\prime}$ is of homogeneity degree $p-M$ in the variables
$\left\{x_{1},..,x_{s}\right\}$ associated to subalgebra
generators and degree $M$ in the remaining variables
$\left\{x_{s+1},..,x_{n}\right\}$. Accordingly, any
$\Phi_{\alpha}$ is of degree $p-\alpha$ in the variables
$\left\{x_{1},..,x_{s}\right\}$ and $\alpha$ in the
$\left\{x_{s+1},..,x_{n}\right\}$. We denote this by saying that
these functions are of bi-degree $(p-\alpha,\alpha)$.

\smallskip

Now the equations (\ref{KS1}) corresponding to subalgebra
generators remain unaltered by the contraction procedure, since
the re-scaling of generators does not affect them.\footnote{On the
contrary, for the remaining equations the differential operators
of the generators corresponding to the representation have been
modified, and the equations are dependent on $t$.} Thus for any
$1\leq i\leq s$ and any homogeneous polynomial $\Psi$ of bi-degree
$(p-q,q)$ we obtain
\begin{equation}
\widehat{X}_{i}\Psi=C_{ij}^{k}x_{k}\frac{\partial \Psi}{\partial
x_{j}}+C_{ij+s}^{k+s}x_{k+s}\frac{\partial \Psi}{\partial x_{j}},
\end{equation}
and the result is easily seen to be again a polynomial with the
same bi-degree. This means that evaluating $C_{p}=t^{M}
C_{p}^{\prime}+ \sum_{\alpha}t^{\alpha} \Phi_{\alpha} +  \Phi_{0}$
is a sum of polynomials of different bi-degree, and since $C_{p}$
is a Casimir operators, the only possibility is that each term is
a solution of the system. We thus conclude that the
$\Phi_{\alpha}$ are solutions of (\ref{KS1}).
\end{proof}

\medskip

The first question that arises from decomposition (\ref{COZ}) is
how many independent additional solutions we obtain. Since all
$\Phi_{\alpha}$ together sum the Casimir operator, some dependence
relations must exist.

\begin{lemma}
Let $C_{p}$ be a Casimir operator of $\frak{s}$ of order $p$.
Suppose that
\begin{equation}
C_{p}=\Phi_{(p-\alpha_{1},\alpha_{1})}+...+\Phi_{(p-\alpha_{q},\alpha_{q})},\quad
0\leq \alpha_{i}<\alpha_{i+1}\leq p
\end{equation}
is the decomposition of $C_{p}$ into homogeneous polynomials of
bi-degree $(p,q)$.
\begin{enumerate}

\item If $\Phi_{(0,p)}\neq 0$, then at most $q-2$
  polynomials $\Phi_{(p-\alpha_{j},\alpha_{j})}$ are functionally
independent on the Casimir operators of $\frak{s}$ and
$\frak{s}^{\prime}$.

\item If $\Phi_{(0,p)}=0$, then at most $q-1$
  polynomials $\Phi_{(p-\alpha_{j},\alpha_{j})}$ are functionally
independent on the Casimir operators of $\frak{s}$ and
$\frak{s}^{\prime}$.
\end{enumerate}
\end{lemma}

The proof follows at once observing that $\Phi_{(0,p)}$ is a
function of the Casimir operators of the subalgebra
$\frak{s}^{\prime}$. The independence on the Casimir operators of
$\frak{s}^{\prime}$ does not imply in general that the
$\Phi_{(p-\alpha,\alpha)}$ obtained are all functionally
independent between themselves. The number of independent terms
depends also on the representation $R$ induced by the reduction
\cite{C23}. In any case, however, at least one independent term is
obtained for any Casimir operator of degree at least three. For
the special case of $n=1$ labelling operator, two terms
independent on the Casimir operators were found, which allowed to
select one as the labelling operator \cite{C72}. Once a set of
functionally independent solutions to system (\ref{KS1}) has been
chosen (including the Casimir operators), the first part of the
labelling problem is solved. Now, if we want to obtain a set of
commuting operators, we have to look for all commutators among the
symmetrized operators $\Phi_{(p-\alpha_{j},\alpha_{j})}$. We
denote by $\Phi_{(p-\alpha_{j},\alpha_{j})}^{symm}$ the
symmetrized polynomial. Then
$\left[\Phi_{(p-\alpha_{j},\alpha_{j})}^{symm},\Phi_{(q-\alpha_{k},\alpha_{k})}^{symm}\right]$
 is a homogeneous polynomial of degree $p+q-1$, and also
constitutes a missing label operator. Actually this brackets is
expressible as sum of polynomials of different bi-degree, and
these terms constitute themselves labelling operators \cite{Jo}. A
procedure to solve the missing label problem can thus be resumed
in the following steps:

\begin{itemize}

\item Decompose the Casimir operators of $\frak{s}$ of degree
$p\geq 3$ with respect to the associated contraction.

\item Determine the commutator of all symmetrized polynomials $
\Phi_{(p-\alpha_{j},\alpha_{j})}^{symm}$ with $\alpha_{j}\neq 0$.

\item From those commuting operators, extract $n$ operators that
are functionally independent from the Casimir operators of
$\frak{s}$ and the subalgebra $\frak{s}^{\prime}$.
\end{itemize}

In general, the second step is reduced to pure computation. There
is no simple procedure to determine whether two missing label
operators are mutually orthogonal, although various symbolic
routines have been developed to compute these brackets (see e.g.
\cite{DM}). In some special circumstances, however, the
decomposition (\ref{COZ}) can provide orthogonality without being
forced to compute the brackets. If for a specific MLP it is known
that no solutions of bi-degree $(r,s)$ exists for some fixed
$r+s=p+q$, and if we have two labelling operators such that
$\left[\Phi_{(p-\alpha_{j},\alpha_{j})}^{symm},\Phi_{(q-\alpha_{k},\alpha_{k})}^{symm}\right]$
is a sum of polynomials of bi-degree $(s,r)$, then the commutation
follows at once. This idea was first explored systematically in
\cite{Jo}. We remark that in the commutative frame, it would
suffice to show that no polynomial function of bi-degree $(r,s)$
is a solution to subsystem (\ref{KS1}).\footnote{The labelling
operator in the enveloping algebra of $\frak{s}$ follows as the
symmetrized form of such functions.}

\section{Examples}

In this section we show how the decomposition of Casimir operators
of higher order provide solutions to missing label problem that
could not be solved completely by only using the contraction, or
for which no proposed set of labelling operators has been computed
yet. We insist on the fact that the main difficulty in the formal
approach to the MLP resides in obtaining a sufficient number of
(functionally) independent labelling operators, from which a
commuting set can be extracted.

\subsection{$G_{2}\supset \frak{su}(2)\times\frak{su}(2)$}

This chain was indicated in \cite{C72} to give an insufficient
number of labelling operators when only the contraction invariants
are considered. Actually, in this case we have
$n=\frac{1}{2}\left(14-2-6-2\right)=2$ labelling operators, and
the inhomogeneous contraction $G_{2}\rightsquigarrow
(\frak{su}(2)\times\frak{su}(2))\overrightarrow{\oplus}_{R}8L_{1}$
preserves the number of invariants. This means that we would only
obtain one additional operator, since the (contracted) operator of
order two is of no use. Now, the method failed because it did not
take into account the decomposition of the sixth order operator
into homogeneous polynomials of bi-degree $(p,q)$ in the
variables. We show that, with this decomposition, we obtain a
complete solution to the MLP related to the chain $G_{2}\supset
\frak{su}(2)\times\frak{su}(2)$. To this extent, we choose the
same tensor basis used in \cite{Hu1} consisting of the generators
$\left\{j_{0},j_{\pm},k_{0},k_{\pm},R_{\mu,\nu}\right\}$ with
$\mu=\pm \frac{3}{2},\pm \frac{1}{2}$, $\nu=\pm\frac{1}{2}$. The
generators $R_{\mu,\nu}$ are related to an irreducible tensor
representation $R$ of $\frak{su}(2)\times\frak{su}(2)$ of order
eight. In this case, the contraction $G_{2}\rightsquigarrow
(\frak{su}(2)\times\frak{su}(2))\overrightarrow{\oplus}_{R}8L_{1}$
is obtained considering the transformations:
\[
j_{0}^{\prime }=j_{0},\;j_{\pm }^{\prime }=j_{\pm
},\;k_{0}^{\prime
}=k_{0},k_{\pm }^{\prime }=k_{\pm },\;R_{\mu ,\nu }^{\prime }=\frac{1}{t}%
R_{\mu ,\nu }.
\]
If we decompose now the Casimir operators $C_{2}$ and $C_{6}$ over
the transformed basis, we obtain the following decomposition
\begin{equation}
\begin{array}[l]{ll}
C_{2}=t^{2}C_{(2,0)}+ C_{(0,2)}, & \\
C_{6}=t^{6}C_{(6,0)}+t^{4}C_{(4,2)}+t^{2}C_{(2,4)}+C_{(0,6)},&
\end{array}
\end{equation}
where  $C_{(0,2)},C_{(0,6)}$  are functions of the Casimir
operators of $\frak{su}(2)\times\frak{su}(2)$. Since $C_{(2,0)}$
is functionally dependent on the invariants of $G_{2}$ and
$\frak{su}(2)\times\frak{su}(2)$, it is not further useful. Now it
can be verified that
\begin{equation}
\frac{\partial{\left(C_{2},C_{6},C_{21},C_{22},C_{(2,4)},C_{(4,2)}\right)}}
{\partial{\left(k_{0},k_{-},j_{0},j_{+},R_{\frac{3}{2},\frac{1}{2}},R_{-\frac{3}{2},\frac{1}{2}}\right)}}\neq
0,
\end{equation}
where $C_{21}$ and $C_{22}$ are the quadratic Casimir operators of
$\frak{su}(2)\times\frak{su}(2)$. This provides us with six
independent operators. A long and tedious computation, due to the
quite high number of terms before and after symmetrization, shows
moreover that the chosen operators commute:
\begin{equation}
\begin{array}[l]{ll}
\left[C_{i},C_{(2,4)}\right]=\left[C_{i},C_{(4,2)}\right]=0, &
i=2,6\\
\left[C_{(4,2)},C_{(2,4)}\right]=0.&
\end{array}
\end{equation}
Therefore the set
$\left\{C_{2},C_{6},C_{21},C_{22},C_{(2,4)},C_{(4,2)}\right\}$ can
be taken to solve the labelling problem.

\smallskip

It should be remarked that a direct comparison with the operators
obtained in \cite{Hu1} is quite difficult, for various reasons. At
first, there the scalars in the enveloping algebra were
considered, not symmetrizations of functions, which implies that
lower order terms where considered when explicitly indicating the
labelling operators. On the other hand, we have only distinguished
the bi-degree, that is, the degree of the polynomials in the
variables of the $\frak{su}(2)\times\frak{su}(2)$ subalgebra and
the tensor representation $R$, while in \cite{Hu1} the order with
respect to any of the copies of $\frak{su}(2)$ was considered,
resulting in operators labelled with three indices. Therefore the
operators $C_{(p,q})$ considered here correspond to the sum of
several scalars there. In addition, our solution contains the term
$C^{(114)}$ excluded in \cite{Hu1},\footnote{This is a scalar
having degree one in each of the copies of $\frak{su}(2)$ and four
in the $R_{\mu,\nu}$ generators.} confirming that the pair of
commuting operators obtained above is different from that found
previously. We also remark that a further distinction of the
degrees of the polynomials $\Phi_{(a,b)}^{symm}$ in the variables
of the $\frak{su}(2)$ copies is not possible due to the
contraction.

\subsection{The chain $\frak{sp}(6)>\frak{su}(3)\times
\frak{u}(1)$}

The unitary reduction of the symplectic Lie algebra of rank three
has found ample applications in the nuclear collective model
\cite{RoG}. In this case, nuclear states are classified by means
of irreducible representations of $\frak{sp}(6)$ reduced with
respect to the unitary subalgebra $\frak{su}(3)\times
\frak{u}(1)$. Since the induced representations are not
multiplicity free, we have to add $n=3$ labelling operators to
distinguish the states. Generating functions for this chain were
studied in \cite{GaK}, but without obtaining explicitly the three
required operators. In this section, we will determine a commuting
set of labelling operators that solves the MLP for this reduction.
As we shall see, this case cannot be solved using only the
invariants of the associated contraction.

\smallskip

We will use the Racah realization for the symplectic Lie algebra
$\frak{sp}\left(6,\mathbb{R}\right)$. We consider the generators
$X_{i,j}$ with $-3\leq i,j\leq 3$ satisfying the condition
\begin{equation}
X_{i,j}+\varepsilon_{i}\varepsilon_{j}X_{-j,-i}=0,
\end{equation}
where $\varepsilon_{i}={\rm sgn}\left( i\right) $. Over this
basis, the brackets are given by
\begin{equation}
\left[  X_{i,j},X_{k,l}\right]  =\delta_{jk}X_{il}-\delta_{il}X_{kj}%
+\varepsilon_{i}\varepsilon_{j}\delta_{j,-l}X_{k,-i}-\varepsilon
_{i}\varepsilon_{j}\delta_{i,-k}X_{-j,l}, \label{Kl3}
\end{equation}
where $-3\leq i,j,k,l\leq 3$. The three Casimir operators
$C_{2},C_{4},C_{6}$ of $\frak{sp}\left( 6,\mathbb{R}\right)  $ are
easily obtained as the coefficients of the characteristic
polynomial
\begin{equation}
\left |  A-T {\rm Id}_{6}\right| =T^{6}+
C_{2}T^{4}+C_{4}T^{2}+C_{6},
\end{equation}
where
\begin{equation}
A=\left(
\begin{array}
[c]{cccccc}%
x_{1,1} & x_{2,1} & x_{3,1} & -I x_{-1,1} & -I x_{-1,2} & - I x_{-1,3}\\
x_{1,2} & x_{2,2} & x_{3,2} & -I x_{-1,2} & -I x_{-2,2} & -I x_{-2,3}\\
x_{1,3} & x_{2,3} & x_{3,3} & -I x_{-1,3} & -I x_{-2,3} & -I
x_{-3,3}\\
I x_{1,-1} & I x_{1,-2} & I x_{1,-3} & -x_{1,1} & -x_{1,2} & -x_{1,3}\\
I x_{1,-2} & I x_{2,-2} & I x_{2,-3} & -x_{2,1} & -x_{2,2}
&-x_{2,3}\\
I x_{1,-3} & I x_{2,-3} & I x_{3,-3} & -x_{3,1} & -x_{2,3}
&-x_{3,3}
\end{array}
\right). \label{M0}
\end{equation}
The symmetrized operators give the usual polynomials in the
enveloping algebra. Since the unitary algebra $\frak{u}(3)$ is
generated by $\left\{X_{i,j} | 1\leq i,j\leq 3\right\}$, in order
to write $\frak{sp}\left( 6,\mathbb{R}\right)$ in a
$\frak{su}(3)\times \frak{u}(1)$ basis, it suffices to replace the
diagonal operators $X_{i,i}$ by suitable linear combinations.
Taking $H_{1}=X_{1,1}-X_{2,2},\; H_{2}=X_{2,2}-X_{3,3}$ and
$H_{3}=X_{1,1}+X_{2,2}+X_{3,3}$ we obtain the Cartan subalgebra of
$\frak{su}(3)$, while $H_{3}$ commutes with all $X_{i,j}$ with
positive indices $i,j$. The invariants over this new basis are
simply obtained replacing the $x_{i,i}$ by the corresponding
linear combinations of $h_{i}$. The contraction $\frak{sp}(6)
\rightsquigarrow (\frak{su}(3)\times
\frak{u}(1))\overrightarrow{\oplus}_{R}12L_{1}$, where $R$ is the
complementary to $({\rm ad}(\frak{su}(3)\otimes (1))$ in the
adjoint representation of $\frak{sp}(6)$:\footnote{More precisely,
$R$ decomposes into a sextet and antisextet with $\frak{u}(1)$
weight $\pm 1$ and a singlet with $\frak{u}(1)$ weight $1$.}
\[
{\rm ad}\frak{sp}(6)= ({\rm ad}\frak{su}(3)\otimes (1))\oplus R.
\]
The contraction is determined by the transformations
\begin{equation}
H_{i}^{\prime}=H_{i},\; X_{i,j}^{\prime}=X_{i,j},\;
X_{-i,j}^{\prime}=\frac{1}{t}X_{-i,j},\;
X_{i,-j}^{\prime}=\frac{1}{t}X_{i,-j},\quad 1\leq i,j\leq 3.
\end{equation}
The contraction $(\frak{su}(3)\times
\frak{u}(1))\overrightarrow{\oplus}_{R}12L_{1}$ satisfies
$\mathcal{N}=3$, thus has 3 Casimir operators that can be obtained
as contraction of $C_{2},C_{4},C_{6}$. Note however that $n=3$,
thus the invariants of the contraction will provide at most two
independent missing label operators. This means that using only
the contraction, we cannot solve the MLP for this chain. In order
to find a third labelling operator, we have to consider the
decomposition of the fourth and sixth order Casimir operators of
$\frak{sp}(6)$. Over the preceding transformed basis we obtain:
\begin{equation}
\begin{array}[l]{ll}
C_{4}= t^{4} C_{(4,0)}+ t^{2} C_{(2,2)}+ C_{(0,4)},& \\
C_{6}= t^{6} C_{(6,0)}+ t^{4} C_{(4,2})+ t^{2} C_{(2,4)}+
C_{(0,6)},&
\end{array}
\end{equation}
where $C_{(k,l)}$ denotes a homogeneous polynomial of $k$ in the
variables of $R$ and degree $l$ in the variables of the unitary
subalgebra. The $C_{(0,k)}$ are functions of the Casimir operators
of $\frak{su}(3)\times \frak{u}(1)$, and therefore provide no
labelling operators. We remark that, before symmetrization,
$C_{(2,2)}$ has $126$ terms, $C_{(2,4)}$ $686$ terms, and
$C_{(4,2)}$ 444 terms. The symmetrized operators
$C_{(2,2)},C_{(4,2)}$ and $C_{(2,4)}$ can be added to the Casimir
operators of $\frak{sp}(6)$ and the subalgebra $\frak{su}(3)\times
\frak{u}(1)$, and the 9 operators can be seen to be (functionally)
independent.
\begin{equation}
\begin{array}[l]{ll}
\left[C_{i},C_{(2,2)}\right]=\left[C_{i},C_{(4,2)}\right]=\left[C_{i},C_{(2,4)}\right]=0,&
i=2,4,6. \\
\left[C_{(2,2)},C_{(4,2)}\right]=\left[C_{(2,2)},C_{(2,4)}\right]=\left[C_{(2,4)},C_{(4,2)}\right]=0.&
\end{array}
\end{equation}

\subsection{Applications to the conformal algebra}

Among the many important problems in Physics where the conformal
group $SO(2,4)$ plays an impotant role, like the dynamical
non-invariance group of hydrogen-like atoms, the application to
the periodic charts of neutral atoms in ions was first considered
in \cite{Bar}. This direction was followed to classify chemical
elements by various authors \cite{Ru}. More recently, the
conformal group and its invariants are in the centre of the more
ambitious program KGR, in order to obtain quantitative predictions
of the periodic table of elements \cite{Ki,Ki2}. To this extent,
the set formed by the three Cartan generators and the Casimir
operators (of degrees $2,3$ and $4$), which commute between
themselves, can be used to label certain physical properties.
However, as noted by Racah \cite{Ra}, this set is still not
sufficient for classification purposes. We have to add three
additional operators\footnote{If $r$ denotes the dimension of a
semisimple Lie algebra $\frak{s}$ and $l$ its rank, the number
$f=\frac{1}{2}\left(r-3l\right)=3$ is usually referred to as the
Racah number.} to obtain a complete set of commuting operators
that solve labelling problems. This follows at once if we consider
the missing label problem for the Cartan subalgebra. In this case
\[
n=\frac{1}{2}\left(15-3-3-3\right)=3.
\]
Therefore the Racah operators can be identified with labelling
operators for the reduction chain determined by the Cartan
subalgebra. To exemplify the procedure, we compute the Racah
operators for the conformal algebra. We use the fact that it is
isomorphic to the Lie algebra $\frak{su}(2,2)$. We start from the
the $\frak{u}(2,2)$-basis formed by the operators $\left\{
E_{\mu\nu},F_{\mu\nu}\right\}  _{1\leq\mu,\nu\leq p+q=n}$ with the constraints%
\begin{eqnarray}
E_{\mu\nu}+E_{\nu\mu}=0,\;F_{\mu\nu}-F_{\nu\mu}=0,\nonumber\\
g_{\mu\mu}=(\left(1,1,-1,-1\right).\nonumber
\end{eqnarray}
The brackets are then given by
\begin{eqnarray}
\left[  E_{\mu\nu},E_{\lambda\sigma}\right]   &
=g_{\mu\lambda}E_{\nu\sigma
}+g_{\mu\sigma}E_{\lambda\nu}-g_{\nu\lambda}E_{\mu\sigma}-g_{\nu\sigma
}E_{\lambda\mu}\label{b1}\\
\left[  E_{\mu\nu},F_{\lambda\sigma}\right]   &
=g_{\mu\lambda}F_{\nu\sigma
}+g_{\mu\sigma}F_{\lambda\nu}-g_{\nu\lambda}F_{\mu\sigma}-g_{\nu\sigma
}F_{\lambda\mu}\\
\left[  F_{\mu\nu},F_{\lambda\sigma}\right]   &
=g_{\mu\lambda}E_{\nu\sigma
}+g_{\nu\lambda}E_{\mu\sigma}-g_{\nu\sigma}E_{\lambda\mu}-g_{\mu\sigma
}E_{\lambda\nu}\label{b2}%
\end{eqnarray}
To recover the conformal algebra, we take the Cartan subalgebra
spanned by the vectors
$H_{\mu}=g_{\mu+1,\mu+1}F_{\mu\mu}-g_{\mu\mu}F_{\mu+1,\mu+1}$ for
$\mu=1..3$. The centre of $\frak{u}(p,q)$ is obviously generated
by $g^{\mu\mu}F_{\mu\mu}$.

\begin{proposition}
A maximal set of independent Casimir invariants of
$\frak{su}\left(  2,2\right)  $ is given by the coefficients
$C_{k}$ of the characteristic polynomial $\left| I A-\lambda {\rm
Id}_{N}\right| =\lambda^{4}+\sum_{k=2}^{4}D_{k}\lambda^{4-k}$,
where
\begin{equation}
A=\left(
\begin{array}
[c]{ccccc}%
-I(\frac{3}{4}h_{1}-\frac{1}{2}h_{2}+\frac{1}{4}h_{3}) &-e_{12}-I
f_{12} & e_{13}+I f_{13} & e_{14}+I f_{14} \\
e_{12}-I f_{12} & I (\frac{1}{4} h_{1}+\frac{1}{2}
h_{2}-\frac{1}{4} h_{3}) & e_{23}+I f_{23} & e_{24}+I
f_{24}\\
e_{13}-I f_{13} & e_{23}-I f_{23} & I (\frac{1}{4}
h_{1}-\frac{1}{2} h_{2}-\frac{1}{4} h_{3}) & e_{34}+I f_{34}\\
e_{14} -I f_{14} & e_{24}-I f_{24} &-e_{34}+I f_{34} & I
(\frac{1}{4}h_{1}-\frac{1}{2}h_{2}+\frac{3}{4}h_{3})
\end{array}
\right). \label{Tv}
\end{equation}
\end{proposition}

The classical Casimir operators are obtained symmetrizing the
functions $C_{k}$. In order to compute the Racah operators, we
consider the MLP for the chain $\frak{h}\subset \frak{su}(2,2)$,
where $\frak{h}$ denotes the Cartan subalgebra. The corresponding
contraction\footnote{In this case, the contraction is no more an
inhomogeneous Lie algebra. The procedure remains however valid,
which suggests that it could also be valid for non-semisimple
algebras.} is defined by the non-singular transformations
\[
H^{\prime}_{i}=\frac{1}{t}H_{i},\quad i=1,2,3.
\]
According to this contraction, the Casimir operators decompose as
follows:
\begin{equation}
\begin{array}[l]{ll}
C_{2}=  t^{2}C_{(2,0)} +C_{(0,2)}, &\\
C_{3}=  t^{3}C_{(3,0)}+ t^{2}C_{(2,1)}+ C_{(0,3)}, &\\
C_{4}=  t^{4}C_{(4,0)}+ t^{3}C_{(3,1)}+t^{2}C_{(2,2)}+C_{(0,4)}, &
\end{array}
\end{equation}
where the $C_{(0,i)}$ are functions of $h_{1},h_{2},h_{3}$. The
functions $I_{ij}$ are all solutions to the MLP. In order to
complete the set of orthogonal operators
$\left\{H_{1},H_{2},H_{3},C_{2},C_{3},C_{4}\right\}$ with three
mutually commuting labelling operators, we first extract those
triples that are functionally independent from the Casimir
operators of $\frak{su}(2,2)$ and the $h_{i}$. We can take for
example $C_{(3,0)},C_{(4,0)},C_{(3,1)}$. Since
\begin{equation}
\frac{\partial{\left(H_{1},H_{2},H_{3},C_{2},C_{3},C_{4},C_{(3,0)},C_{(4,0)},C_{(3,1)}\right)}}
{\partial{\left(h_{1},h_{2},h_{3},e_{12},e_{13},e_{14},f_{23},f_{24},f_{34}\right)}}\neq
0,
\end{equation}
these operators are independent. A somewhat more laborious
computation shows that the symmetrized forms of
$C_{(2,1)},C_{(4,0)},C_{(3,1)}$ satisfy the commutators
\begin{equation}
\begin{array}[l]{ll}
\left[C_{i},C_{(3,0)}\right]=\left[C_{i},C_{(4,0)}\right]=
\left[C_{i},C_{(3,1)}\right]=0, & i=1,2,3.\\
\left[C_{(3,0)},C_{(4,0)}\right]=\left[C_{(3,0)},C_{(3,1)}\right]
=\left[C_{(4,0)},C_{(3,1)}\right]=0.&
\end{array}
\end{equation}

Since the $C_{(i,j)}$ are solutions to the MLP, they commute with
any function of the generators $H_{i}$. In contrast to the
previous cases, the length of the labelling operators never
exceeds $70$ terms before symmetrization. In conclusion, the set
$\left\{H_{1},H_{2},H_{3},C_{2},C_{3},C_{4},C_{(3,0)},C_{(4,0)},C_{(3,1)}\right\}$
is complete formed by commuting operators. We observe that linear
combinations of the three Racah operators are of potential use to
describe chemical and physical properties like ionization energy,
atomic volume or magnetic properties. Since this identification
relies heavily on experimental data \cite{Ru}, it remains to
compute the corresponding eigenvalues for irreducible
representations (IRREPs) of $\frak{su}(2,2)$, which constitutes a
quite hard numerical problem. This task is in progress.

\section{Conclusions}

The method of contraction is useful to solve the MLP when the
number of invariants of the contraction associated to the
reduction chain $\frak{s}\supset \frak{s}^{\prime}$ exceeds the
number of needed labelling operators. In the case where the
invariants of the inhomogeneous contraction do not suffice to find
a complete solution of the missing label problem, it is expectable
that labelling operators of the same degree appear. This suggests
that further terms of the Casimir operators of $\frak{s}$ that
disappear during the contraction can be useful to complete the set
of missing label operators. We have shown that the contraction
induces a decomposition of the Casimir operators, the terms of
which are all solutions to the MLP. From these terms a set of $n$
independent labelling operators can be extracted, reducing the
problem to determine which combinations are mutually orthogonal.
In this sense, the method proposed in \cite{C72} is
 a first approximation to solve the MLP using the properties
of reduction chains, which however turns out to be useful in most
practical cases. The bi-degree of the Casimir operators of a Lie
algebra with respect to the variables associated to the generators
of a subalgebra are therefore a relevant tool to obtain and
classify these labelling operators, although further distinction
of terms, for example when the subalgebra consists of various
copies, is also convenient to deduce additional
operators.\footnote{This turns out to be the case for the chain
$\frak{su}(4)\supset \frak{su}(2)\times\frak{su}(2)$ \cite{El}.}
This subdivision cannot however be deduced from the contraction,
since all generators of the subalgebra play the same role.

\medskip

Some important aspects of the decomposition method of Casimir
operators based on the contractions and its use in labelling
problems are specially emphasized:

\begin{itemize}
\item  The solutions provide a  ``natural"  choice for the
labelling operators. Their interpretation as ``broken" Casimir
operators confers them a certain physical meaning, in contrast to
operators obtained by pure algebraic means, where the physical
interpretation of the operator is often not entirely clear.

\item  The decomposition provides also a consistent explanation to
the question why a number of reduction chains give labelling
operators of the same degree. This fact is directly related to an
insufficient number of invariants in the contraction associated to
the chain.

\item This could probably explain why the eigenvalues of such
labelling operators are not integers, as already indicated by
Racah \cite{Ra}. It follows from the decomposition that the
eigenvalues of the labelling operators contribute to the
eigenvalues of the Casimir operators. In this context, the
interpretation of a labelling operator as ``broken" Casimir
operator leads to the idea of ``broken" integer eigenvalues.
\end{itemize}

Some questions still remain open, namely, whether there exist
reductions $\frak{s}\supset \frak{s}^{\prime}$ for which the
method followed here provides all available labelling operators.
An answer in this direction implies to find the general solution
to the MLP for each considered chain. Nowadays, only for a few
number of algebras these computations have been carried out
completely \cite{El,Pa}. A complete study of all physically
relevant reduction chains involving simple Lie algebras up to some
fixed rank would certainly provide new insights to this problem.
On the other hand, in can also not be excluded that for reduction
chains with a great number of labelling operators, the terms of
the decomposition are not sufficient to construct a set of
independent labelling operators. To which extent the invariants of
the contraction not appearing as contracted operators play a role
must still be analyzed.\footnote{Such situations appear, e.g.,
considering a simple algebra of high rank and regular subalgebras
of low rank. If the induced representation (\ref{IRD}) contains
copies of the trivial representation, then the generators
associated to these will play the role of labelling operators. It
happens moreover that these generators cannot be obtained
contracting the Casimir operators of $\frak{s}$. This situation is
however unlike to appear in some physically interesting
case.}\newline Another problem, still in progress, is to obtain
complete sets of commuting operators for all simple Lie algebras,
using the MLP determined by the Cartan subalgebra. The commented
application to the periodic charts of atoms in only one of the
problems where this special type of reductions have been shown to
be of interest in developing algebraic models in molecular physics
or nuclear spectroscopy \cite{Oss}.

\subsection*{Acknowledgments}
This work was partially supported by the research project
MTM2006-09152 of the Ministerio de Educaci\'on y Ciencia. The
author expresses his gratitude to M. R. Kibler and Ch. Quesne for
useful discussions and remarks, as well as to G. Vitiello and J.
Van der Jeugt for additional references.


\section*{References}
\begin{thebibliography}{9}

\bibitem{Ok} Okubo S 1962 \textit{Prog. Theor. Phys.} \textbf{27}
949\nonum Goldberg H and Lehrer-Ilamed Y 1962 \textit{J. Math.
Phys.} \textbf{4} 501 \nonum Burakovsky L and Horwitz L P 1997
\textit{Found. Phys. Lett.} \textbf{10} 131 \nonum Burakovsky L
and Goldman T 1997 (hep-ph/9708498)

\bibitem{Sua} S\"ualp G and Kaptanoglu S 1983 \textit{Ann. Phys.}
\textbf{147} 460 \nonum Kitazawa N 1994 \textit{Tumbling and
Technicolor Theory} DPNU-94-03 (Nagoya: Japan).

\bibitem{Ia} Iachello F and Arima A 1987 {\it The interacting boson
model}, (Cambridge Univ. Press, Cambridge).

\bibitem{El} Elliott J P 1958 \textit{Proc. Roy. Soc. Lond. A}
\textbf{245} 128, 562 \nonum Bargmann V and Moshinsky M 1965 \NP
\textbf{23} 177 \nonum Green H S and Bracken A J 1974 \textit{Int.
J. Theor. Phys.} \textbf{11} 157 \nonum Judd B R, Miller W, Patera
J and Winternitz P 1974 \JMP \textbf{15} 1787 \nonum Sharp R T
1975 \textit{J. Math. Phys.} \textbf{16}, 2050 \nonum Quesne Ch
1976 \textit{J. Math. Phys.} \textbf{17}, 1452;\; 1977
\textbf{18}, 1210 \nonum De Meyer H, Vanden Berghe G, Van der
Jeugt J and De Wilde P 1985 \JMP \textbf{26} 2124 \nonum Rowe D J,
Le Blanc R and Repka J 1989 \JPA \textbf{22} L309 \nonum Hecht K T
1994 \JPA \textbf{27} 3445\nonum Turner P S, Rowe D J and Repka J
2006 \textit{J. Math. Phys.} \textbf{47} 023507

\bibitem{Pe} Peccia A and Sharp R T 1976
\textit{J. Math. Phys.} \textbf{17} 1313

\bibitem{Is} Wigner E P 1937 \textit{Phys. Rev.} \textbf{51}
106\nonum Lipkin H J 1964 \textit{Phys. Rev. Lett.} \textbf{13}
590 \nonum Khanna F C and Umezawa H 1994 \textit{Chinese J. Phys.}
\textbf{32} 1317 \nonum von Isacker P and Juillet O 1999
\textit{Nuclear Phys.} \textbf{A564} 739\nonum Valencia J P and Wu
H C 2004 \textit{Braz. J. Phys.} \textbf{34} 837

\bibitem{C72} Campoamor-Stursberg 2007 \textit{J. Phys. A: Math. Theor.}, to appear
(hep-th/arXiv:0706.2581)

\bibitem{Vi} Vitiello G and de Concini C 1976 Nuclear Phys. B{\bf
116} 141 \nonum Anderson J T 1981 \textit{Phys. Rev.} \textbf{D23}
1856 \nonum Celeghini E, Tarlini M and  Vitiello G 1984
\textit{Nuovo Cimento A} \textbf{84} 19

\bibitem{C49}  Campoamor-Stursberg R 2006 \JPA \textbf{39}
2325 \nonum \dash 2007 \textit{J. Phys. A: Math. Theor.}
\textbf{40} 5355

\bibitem{IW} Segal I E 1951 \textit{Duke Math. J.} \textbf{18}
221 \nonum In\"on\"u E and Wigner E P 1953 \textit{Proc. Nat.
Acad. Sci U.S.A.} \textbf{39} 510\nonum Domokos G and Tindle G L
1968 \textit{Commun. Math. Phys.} \textbf{7} 160

\bibitem{C23} Campoamor-Stursberg R 2003 \JPA \textbf{36} 1357

\bibitem{Jo} Van der Jeugt 1984 \textit{J. Math. Phys.}
\textbf{25} 1221

\bibitem{DM} De Meyer H, Vanden Berghe G and De Wilde P 1987
\textit{Comput. Phys. Comm.} \textbf{44} 197

\bibitem{Hu1} Hughes J W B and Van der Jeugt J 1985 \JMP
\textbf{26} 894

\bibitem{RoG} Rosensteel G 1977 \textit{Int. J. Theor. Phys.}
\textbf{16} 63\nonum Graber J L and Rosensteel G 2003
\textit{Phys. Rev. C} \textbf{68} 014301

\bibitem{GaK} Gaskell R, Rosensteel G and Sharp R T 1981
\textit{J. Math. Phys.} \textbf{22} 2732

\bibitem{Bar} Barut A O 1972 in {\it The Structure of Matter}
(Proc. Rutherford Centennial Symp., 1971) (Christchurch, New
Zealand)

\bibitem{Ru} Rumer Yu B and Fet A I 1971 \textit{Teor. Mat. Fiz.}
\textbf{9} 203 \nonum Konopel'chenko B G and Rumer Yu B 1979
\textit{Uspekhi Fiz. Nauk} \textbf{129} 339\nonum Byakov V M,
Kulakov Yu I, Rumer Yu B and Fet A I 1976 Preprint ITEP-26,
ITEP-90\nonum Kibler M R 1989 \textit{J. Mol. Struct.}
\textbf{187} 83\nonum Carlson C M, Hefferlin R A and Zhuvikin G V
1995 \textit{Analysis of Group Theoretical Periodic Systems of
Molecules using Tabulated Data} Joint Report SC/SPBU-2 \nonum
Kibler M R 2004 \textit{The Mathematics of the Periodic Table}, ed
D H Rouvray and R B King (Baldock: Research Studies Press)

\bibitem{Ki} Kibler M R 2004 \textit{Mol. Phys.} \textbf{102} 1221

\bibitem{Ki2} Kibler M R, private communication

\bibitem{Ra} Racah G 1949 {\it Phys. Rev.} \textbf{76} 1352 \nonum
\dash 1951 {\it Group Theory and Spectroscopy}, (Princeton Univ.
Press, New Jersey) \nonum Fano U and Racah G 1959
\textit{Irreducible tensorial Sets} (Academic Press: New York)

\bibitem{Pa} Partensky A and Maguin C 1978 \textit{J. Math. Phys.}
\textbf{19} 511

\bibitem{Oss} Frank A and van Isacker P 1994 \textit{Algebraic
Methods in Molecular and Nuclear Structure Physics} (Wiley: New
York) \nonum Iachello F and Levine R D 1995 \textit{Algebraic
Theory of Molecules} (Oxford Univ. Press: Oxford)

\end{thebibliography}
\end{document}